\documentclass[a4paper,UKenglish]{llncs}
\usepackage{stmaryrd}
\usepackage{mathdots}

\usepackage{float}
\usepackage{multirow}
\usepackage{here}
\usepackage{ae,aeguill}
\usepackage{fancyhdr}
\usepackage{amssymb}
\usepackage{wrapfig}
\usepackage{amsmath}
\usepackage{xspace}
\usepackage{colordvi}
\usepackage{graphicx}
\usepackage{amsfonts}
\usepackage[toc,page]{appendix} 
\usepackage{fancybox}
\usepackage{moreverb}
\usepackage{pgf}
\usepackage{tabularx}
\usepackage{thm-restate}
\usepackage{tikz}
\usepackage{subfig}
\usepackage[utf8]{inputenc}

\usepackage{bbm}

\usetikzlibrary{calc,automata,arrows,chains,matrix,positioning,scopes}
\makeatletter

\@ifundefined{c@subfigure}{\newsubfloat{figure}}{}

\renewcommand{\leq}{\leqslant}

\newcommand{\cA}{\mathcal{A}}
\newcommand{\cB}{\mathcal{B}}
\newcommand{\cX}{\mathcal{X}}

\newcommand{\cP}{\mathcal{P}}

\newcommand{\LL}{\sim_{ll}}
\newcommand{\LR}{\sim_{lr}}
\newcommand{\RL}{\sim_{rl}}
\newcommand{\RR}{\sim_{rr}}

\newcommand{\calS}{\mathcal{S}}

\newcommand{\acrostyle}[1]{\ensuremath{\mathrm{#1}}\xspace}
\newcommand{\NFT}{\acrostyle{NFT}}
\newcommand{\NFTs}{\acrostyle{NFTs}}
\newcommand{\DFT}{\acrostyle{DFT}}

\newcommand{\TWNFT}{\acrostyle{2NFT}}

\newcommand{\TWDFT}{\acrostyle{2DFT}}
\newcommand{\twdft}{\TWDFT}

\newcommand{\NFA}{\acrostyle{NFA}}

\newcommand{\DFA}{\acrostyle{DFA}}

\newcommand{\TWDFA}{\acrostyle{2DFA}}

\newcommand{\sst}{\acrostyle{SST}}

\newcommand{\FO}{\acrostyle{FO}}
\newcommand{\FOT}{\acrostyle{FOT}}
\newcommand{\MSO}{\acrostyle{MSO}}
\newcommand{\MSOT}{\acrostyle{MSOT}}

\newcommand{\stm}{\acrostyle{STM}}
\newcommand{\ftm}{\acrostyle{FTM}}

\newcommand{\Vars}{\mathcal{X}}
\newcommand{\seq}[1]{\langle #1 \rangle}

\newcommand{\map}{\phi}
\newcommand{\inp}{\textsf{in}}
\newcommand{\outp}{\textsf{out}}

\newcommand{\prop}[1]{\textbf{\textsf{P}$_{#1}$}}
\newcommand\ins{\mathbbmtt{i}}
\newcommand\outs{\mathbbmtt{o}}

\newcommand{\tsig}{\tilde{\sigma}}

\newcommand{\bh}{\operatorname{bh}}

\title{Aperiodic String Transducers~\thanks{This work is supported by
the ARC project Transform (French speaking community of Belgium), the
Belgian FNRS PDR project Flare,
and the
PHC project VAST (35961QJ) funded by Campus France and WBI.}}
\author{Luc Dartois\inst{1}, Ismaël Jecker\inst{1}, Pierre-Alain Reynier\inst{2}}

\authorrunning{Dartois, Jecker, Reynier}

\institute{Université Libre de Bruxelles, Belgium
\and Aix-Marseille Université, CNRS, LIF UMR 7279, France}
  	\date{\today}

\begin{document}
\maketitle

\vspace{-.5cm}


\begin{abstract}
Regular string-to-string functions enjoy a nice triple characterization through deterministic two-way transducers (\twdft), streaming string transducers (\sst) and MSO definable functions.
This result has recently been lifted to FO definable functions, with equivalent representations
by means of \emph{aperiodic} \twdft
and \emph{aperiodic} 1-bounded
\sst, extending a well-known result on regular languages. 
In this paper, we give three direct transformations: $i)$ from 1-bounded \sst to \twdft,
$ii)$ from \twdft to copyless \sst, and
$iii)$ from $k$-bounded to $1$-bounded \sst.
We give the complexity of each construction and also prove that they preserve the aperiodicity of transducers.
As corollaries, we obtain that FO definable string-to-string
functions are equivalent to \sst whose transition
monoid is finite and aperiodic, and to aperiodic copyless \sst.
\end{abstract}

\section{Introduction}


The theory of regular languages constitutes a cornerstone in theoretical computer science. 
Initially studied on languages of finite words, it has since been extended in numerous directions,
 including finite and infinite trees. Another natural extension is moving from languages to 
 transductions. We are interested in this work in string-to-string transductions, and more precisely in string-to-string functions. One of the strengths of the class of regular languages is their equivalent
  presentation by means of automata, logic, algebra and regular expressions. The class of so-called \emph{regular string functions} enjoys a similar multiple presentation. 
  It can indeed be alternatively defined using deterministic two-way finite state transducers
   (\TWDFT), using Monadic Second-Order graph transductions interpreted 
   on strings (\MSOT)~\cite{EH01}, and using the model of 
   streaming string transducers (\sst)~\cite{AC10}. 
    More precisely, regular string functions are equivalent to different classes of \sst, namely 
 copyless \sst~\cite{AC10} and $k$-bounded \sst, for every positive integer $k$~\cite{AFT12}. 
   Different papers~\cite{EH01,AC10,AFT12,ADT13} have proposed transformations 
   between \TWDFT, \MSOT and \sst, summarized on 
   Figure~\ref{fig:overview}.


The connection between automata and logic, which has been very fruitful for model-checking for instance, also needs to be investigated in the framework of transductions. As it has been done for regular languages, an important objective is then to provide similar logic-automata connections for subclasses of regular functions, providing decidability results for these subclasses. As an illustration, the class of rational functions (accepted by one-way finite state transducers) owns a simple characterization in terms of logic, as shown in~\cite{Filiot-ICLA15}. The corresponding logical fragment is 
  called order-preserving \MSOT. The decidability of the one-way definability of a two-way transducer proved in~\cite{FGRS-LICS13} thus yields the decidability of this fragment 
  inside the class of \MSOT.

\begin{figure}[t]
\begin{tikzpicture}[->,>=latex,node distance=3cm]
 \begin{scope}
  \node (1sst) at (0,0) {(aperiodic) $1$-b. \sst};
  \node (mso) at (4.5,-3) {(\FOT) \MSOT};
  \node (d1) at (3.55,-2.9) {};
  \node (d2) at (-3.3,-2.9) {};
  \node (d3) at (3.55,-3.1) {};
  \node (d4) at (-3.3,-3.1) {};
  \node (2w) at (-4.5,-3) {(aperiodic) \TWDFT};
  \node (csst) at (-4.5,0) {(aperiodic) copyless \sst};
  \node (ksst) at (4.5,0) {(aperiodic) $k$-b. \sst};
  \path[thick] 
  (csst) edge [bend right=5, above, pos=.55] node {\small \cite{AC10,AFT12}} (mso)
  (csst) edge [ultra thick, bend left=15,right,pos=.5,align=center,dashed] (2w)
  (ksst) edge [ultra thick, bend left=25,below,pos=.5,dashed] ([xshift=1.3cm,yshift=-.2cm]1sst)
  (ksst) edge [bend right=15,above,pos=.5] node {\small \cite{AFT12}} (csst)
  (2w) edge [below, bend right=15] node {\small \cite{AFT12}} (1sst)
  (2w) edge [ultra thick,dashed,bend left=15,left,pos=.5,align=center] node {\small \cite{AC10}} (csst)
  ([xshift=-2cm] mso) edge [bend left=5, below,pos=.6] node {\small \cite{ADT13}} (1sst)
  ([xshift=.7cm] mso) edge [<->,above,pos=.4, dotted, thick] node {\small \cite{FKT14}} ([xshift=.7cm,yshift=-.2cm] 1sst)
  (d1) edge [<->,above,pos=.5,align=center,thick] node {\small \cite{EH01}} (d2)
  (d3) edge [<->,below,pos=.5,align=center,dotted,thick] node {\small \cite{CD15}} (d4)
  (1sst) edge [color=black!0] node {\color{black}{$\subseteq$}} (ksst)
  (1sst) edge [ultra thick,dashed,bend right=15] (2w)
  (csst) edge [color=black!0] node {\color{black}{$\subseteq$}} (1sst)
  ;
  
  \end{scope}
\end{tikzpicture}
%
%
\caption{{\small Summary of transformations between equivalent models. $k$-b. 
stands for $k$-bounded. Plain (resp. dotted) arrows concern regular models (resp. bracketed models).  
Original constructions presented in this paper are depicted by thick dashed arrows
and are valid for both regular and aperiodic versions of the models.
}}\label{fig:overview}
\end{figure}
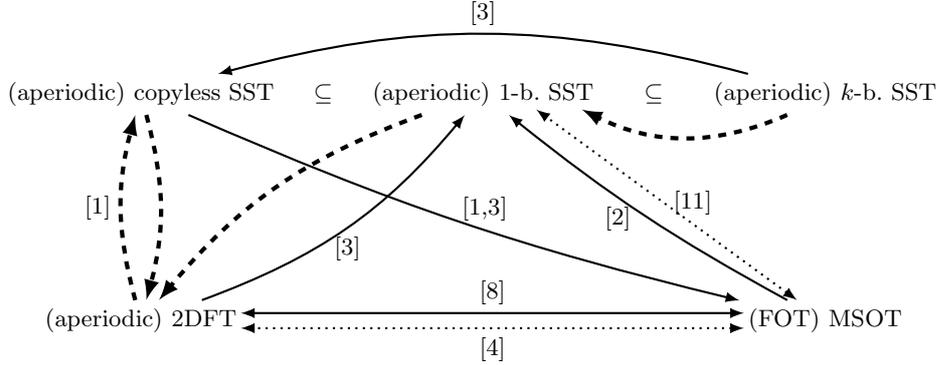

 
 The first-order logic considered with order predicate constitutes an important 
 fragment of the monadic second order logic. It is well known that languages definable 
 using this logic are equivalent to those recognized by finite state automata whose 
 transition monoid is aperiodic (as well as other models such as star-free regular expressions). 
 These positive results have motivated the study of similar connections between first-order 
 definable string transformations (\FOT) and restrictions of state-based transducers models. 
 Two recent works provide such characterizations 
 for $1$-bounded \sst and \TWDFT respectively~\cite{FKT14,CD15}.  To this end, the authors study a 
 notion of transition monoid for these transducers, and prove that \FOT is expressively 
 equivalent to transducers whose transition monoid is aperiodic by providing back and 
 forth transformations between \FOT and $1$-bounded aperiodic \sst (resp. aperiodic \TWDFT). 
In particular, \cite{FKT14} lets as an open problem whether \FOT is also 
equivalent to aperiodic copyless \sst and to aperiodic $k$-bounded \sst, for every positive integer $k$.
 It is also worth noticing that 
  these characterizations of \FOT, unlike the case of languages, do not allow to decide the class \FOT inside the class \MSOT. Indeed, while decidability for languages relies on the syntactic congruence of the language, no such canonical object exists for the class of regular string transductions.


%
%

In this work, we aim at improving our understanding of the relationships between \TWDFT and \sst. 
We first provide an original transformation from $1$-bounded (or copyless) \sst to \TWDFT, and study its complexity.
While the existing construction used \MSO transformations as an intermediate formalism, resulting in a non-elementary
complexity, our construction is in double exponential time, and in single exponential time
if the input \sst is copyless.
 Conversely, we describe a direct construction from \TWDFT to copyless \sst, which is similar 
 to that of~\cite{AC10}, but avoids the use of an intermediate model.
  These constructions also allow to establish links between the crossing degree of a \TWDFT, and 
the number of variables of an equivalent copyless (resp. $1$-bounded) \sst, and conversely. 
Last, we provide a direct construction from $k$-bounded \sst to $1$-bounded \sst, while the existing one was using copyless \sst as a target model and not $1$-bounded \sst~\cite{AFT12}. 
These constructions are represented by thick dashed arrows on Figure~\ref{fig:overview}.


In order to lift these constructions to aperiodic transducers, we introduce a new transition monoid for \sst,
which is intuitively more precise than the existing one (more formally, the existing one divides the one we introduce).
We use this new monoid to prove that the three constructions we have considered above 
preserve the aperiodicity of the transducer. As a corollary,  this implies that \FOT
is equivalent to both aperiodic copyless and $k$-bounded \sst, for every integer $k$, two results 
that were
stated as conjectures in~\cite{FKT14} (see Figure~\ref{fig:overview}).

\section{Definitions}

\subsection{Words, Languages and Transducers}
Given a finite alphabet $A$, we denote by
$A^*$ the set of finite words over $A$, and by
$\epsilon$ the empty word.
The length of a word $u\in A^*$ is its number of
symbols, denoted by $|u|$. For all $i\in\{1,\dots,|u|\}$, we
denote by $u[i]$ the $i$-th letter of $u$.

A \emph{language} over $ A$ is a set $L\subseteq  A^*$. Given two alphabets $A$ and $B$, 
a \emph{transduction} from $ A$ to $B$ is a relation 
$R \subseteq  A^*\times  B^*$. 
 A transduction $R$ is
\emph{functional} if it is a function.
The transducers we will introduce will define transductions. We will say that
two transducers $T,T'$ are equivalent whenever they define the same transduction.

\noindent\textbf{Automata}
A \emph{deterministic two-way finite state automaton} (\TWDFA) over a finite alphabet
$A$ is a tuple $\cA = (Q, q_0, F, \delta)$  where $Q$ is a finite set of states, $q_0\in Q$ is 
the initial state, $F\subseteq Q$ is a set of final states, and $\delta$ is the
transition function, of type $\delta: Q\times ( A\uplus\{\vdash,\dashv\})\to  Q\times \{+1,0,-1\}$. 
The new symbols $\vdash$ and $\dashv$ are called \emph{endmarkers}.

An input word $u$ is given enriched by the endmarkers, meaning that $\cA$ reads the input $\vdash u\dashv$. We set $u[0]= \vdash$ and $u[|u|+1]=\dashv$.
Initially the head of $\cA$ is on the first cell $\vdash$ in state $q_0$ (the cell at position $0$). 
When $\cA$ reads an input symbol, depending on the
transitions in $\Delta$, its head moves to the left ($-1$), 
or stays at the same position ($0$),
or moves to the right ($+1$).
To ensure the fact that the reading of $\cA$ does not go out of bounds,
we assume that there is no transition moving to the left (resp. to the right) 
on input symbol $\vdash$ (resp. $\dashv$).
 $\cA$ stops as soon as it reaches the endmarker $\dashv$ in a final state.

A \emph{configuration} of $\cA$ is a pair $(q,i)\in Q\times \mathbb{N}$ where $q$ is a state and $i$ is a position on the input tape.
A \emph{run} $\rho$ of $\cA$ is a finite sequence of configurations. 
The run $\rho = (p_1,i_1)\dots (p_m,i_m)$ is a run 
on an input word $u\in A^*$ of length $n$ if $i_m\leq n+1$, and for all $k\in \{1,\dots,m-1\}$, 
$0\leq i_k\leq n+1$ and $(p_k,u[i_k], p_{k+1}, i_{k+1}-i_k)\in \Delta$. It is \emph{accepting} if $p_1 = q_0$, $i_1 = 0$, and $m$ is the only index where both $i_m = n+1$ and $p_m \in F$.
The language of a \TWDFA $\cA$, denoted by $L(\cA)$, is the set of words $u$ such that there exists an accepting run
of $\cA$ on $u$.


\noindent\textbf{Transducers}
\emph{Deterministic two-way finite state transducers} (\TWDFT) over $ A$ extend \TWDFA with a one-way left-to-right output tape.
They are defined as \TWDFA except that the transition relation $\delta$ is extended with outputs:
$\delta: Q\times ( A\uplus\{\vdash,\dashv\}) \to B^* \times Q \times \{-1,0,+1\}$. When 
a transition $(q,a,v,q',m)$ is fired, 
the word $v$ is appended to the right of the output tape.

A run of a \TWDFT is a run of its underlying automaton, i.e. the \TWDFA obtained by ignoring the output 
(called its \emph{underlying input automaton}).
A run $\rho$ may be simultaneously a run on a word $u$ and on a word
$u'\neq u$. However, when the input word is given, there is
a unique sequence of transitions associated with $\rho$. Given a
\TWDFT $T$, an input word $u\in A^*$ and a run $\rho =
(p_1,i_1)\dots (p_m,i_m)$ of $T$ on $u$, the output of $\rho$ on $u$ 
is the word obtained by concatenating the outputs of
the transitions followed by $\rho$. If $\rho$ contains a single configuration, 
this output is simply $\epsilon$.
The transduction defined by $T$ is the relation $R(T)$
defined as the set of pairs $(u,v)\in A^*\times B^*$ 
such that $v$ is the output of an accepting run $\rho$
on the word $u$. As $T$ is deterministic, such a run 
is unique, thus $R(T)$ is a function.

\begin{figure}[t!]
\begin{center}
 \begin{minipage}[c]{0.3\linewidth}
	  \begin{tikzpicture}[scale=0.70, initial text=,inner sep=0pt]
  \node[state,initial,accepting where=left, initial where=left] (q) at (0,0) {$1$};
  \node[state] (r) at (3,0) {$2$};
  \node[state, accepting,accepting where=left] (s) at (6,0) {$3$};
  \path[->] (q) edge [loop above]   node   {\small{$a|a,+1$}}();
  \node () at (-0.04,1.84) { \small{${\vdash}|\epsilon,+1$}};
  \path[->] (q) edge [above] node {\small{$b|\epsilon,-1$}}(r);
  \node () at (1.5,0.75) { \small{${\dashv}|\epsilon,-1$}};
  \path[->] (r) edge [loop above]   node   {\small{$a|b,-1$}}();
  \path[->] (r) edge [above] node {\small{$b|\epsilon,+1$}} (s);
  \node () at (4.5,0.75) {\small{${\vdash}|\epsilon,+1$}};
  \path[->] (s) edge [loop above] node {\small{$a|\epsilon,+1$}} ();
  \path[->] (s) edge [bend left=30,below] node {\small{$b|\epsilon,+1$}} (q);
		  \end{tikzpicture} 
 	 \end{minipage}\hfill
 \begin{minipage}[l]{0.3\linewidth}
 \begin{tikzpicture}[initial text=,scale=0.7]
 
 \node[state,initial above] (state) at (0,0) {};
  \path[->,align=center] (state) edge [loop left] node {
  \begin{tabular}{c|c}
\multirow{2}{6pt}{$a$} & $X=Xa$ \\
						& $Y=Yb$
  \end{tabular}
	 		} ();
  \path[->] (state) edge [loop right] node {
    \begin{tabular}{c|l}
\multirow{2}{6pt}{$b$} & $X=XY$ \\
						& $Y=\epsilon$
  \end{tabular}
  } ();
  \path[->] (state) edge [right] node {\small $XY$} (0,-1.3);
 
 \end{tikzpicture}
 \end{minipage}\hfill
  	 \begin{minipage}{0.1\linewidth}
 	 \end{minipage}\hfill
\end{center}
\caption{Aperiodic \TWDFT (left) and \sst (right) realizing the function $f$.}\label{Fig:example}
\end{figure}
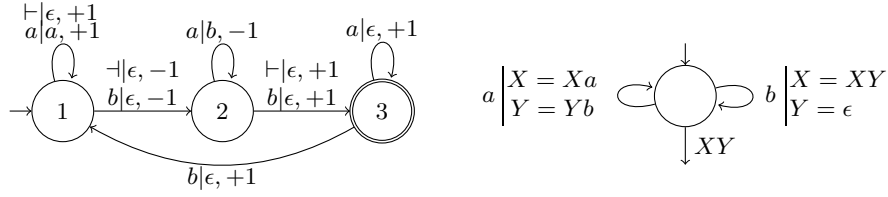

\noindent\textbf{Streaming String Transducers}
Let $\Vars$ be a finite set of variables denoted by $X,Y,\dots$ and $B$ be a finite alphabet. 
A substitution $\sigma$ is defined as a mapping ${\sigma : \Vars \to (B \cup
  \Vars)^*}$. Let $\mathcal{S}_{\Vars, B}$ be the set of all substitutions.
Any substitution $\sigma$ can be extended to $\hat{\sigma}: (B \cup \Vars)^*
\to (B \cup \Vars)^*$ in a straightforward manner.
The composition $\sigma_1 \sigma_2$  of two substitutions $\sigma_1$ and $\sigma_2$
is defined as the standard function composition  $\hat{\sigma_1} \sigma_2$,
i.e. $\hat{\sigma_1}\sigma_2(X) = \hat{\sigma_1}(\sigma_2(X))$ for all $X \in \Vars$. 
We say that a string $u \in (B \cup \Vars)^*$ is \emph{$k$-linear} if each 
$X \in \Vars$ occurs at most $k$ times in $u$. 
A substitution $\sigma$ is $k$-linear if $\sigma(X)$
is \emph{$k$-linear} for all $X$. 
It is \emph{copyless} if for any variable $X$, there exists at most one variable $Y$ such that $X$ occurs in $\sigma(Y)$, and $X$ occurs at most once in $\sigma(Y)$.

  A \emph{streaming string transducer} (\sst) is a tuple 
  $T = ( A, B, Q, q_0, Q_f, \delta, \Vars, \rho, F)$ where 
  $(Q,q_0,Q_f,\delta)$ is a one-way automaton, 
  $A$ and $B$ are finite sets of input and output alphabets respectively,
  $\Vars$ is a finite set of variables,
  $\rho : \delta \to \mathcal{S}_{\Vars, B}$ is a variable update and
  $F: Q_f \rightharpoonup (\Vars\cup B)^*$ is the output function.


\begin{example}\label{example}
As an example, let $f:\{a,b\}^*\to\{a,b\}^*$ be the function mapping any word $u = a^{k_0}ba^{k_1}\cdots ba^{k_n}$ to the word $f(u) =
  a^{k_0}b^{k_0}a^{k_1}b^{k_1} \cdots a^{k_n}b^{k_n}$ obtained by adding
  after each block of consecutive~$a$ a block of consecutive~$b$ of the
  same length.  Since each word $u$ over~$A$ can be uniquely written $u =
  a^{k_0}ba^{k_1} \cdots ba^{k_n}$ with some $k_i$ being possibly equal to $0$, the function~$f$ is well defined.
  We give in Figure~\ref{Fig:example} a \TWDFT and an \sst that realize $f$.
\end{example}

The concept of a run of an \sst is defined in an analogous manner to that of
a finite state automaton.
The sequence $\seq{\sigma_{r, i}}_{0 \leq i \leq |r|}$ of substitutions induced
by a run $r = q_0 \xrightarrow{a_1} q_1 \xrightarrow{a_2} q_2 \ldots q_{n-1}\xrightarrow{a_n} q_n$ is defined
inductively as the following: $\sigma_{r, i} {=} \sigma_{r, i{-}1} \rho(q_{i-1}, a_{i})$
for $1 < i \leq |r|$ and $\sigma_{r,1} = \rho(q_0,a_1)$. We denote $\sigma_{r,|r|}$ by $\sigma_r$
and say that $\sigma_r$ is induced by $r$.

If $r$ is accepting, i.e. $q_n\in Q_f$, we can extend the output function $F$ to $r$ by 
$F(r) = \sigma_\epsilon\sigma_{r}F(q_n)$, where $\sigma_\epsilon$
substitutes all variables by their initial value $\epsilon$. For all
words $u\in A^*$, the output of $u$ by $T$ is defined only if 
there exists an accepting run $r$ of $T$ on $u$, and in that case the output
is denoted by $T(u) = F(r)$. 
The transformation $R(T)$ is then defined as the set of
pairs $(u,T(u))\in A^*\times B^*$.

An \sst $T$ is copyless if for every transition $t\in\delta$, the variable
update $\rho(t)$ is copyless. Given an integer $k\in \mathbb{N}_{>0}$,
we say that $T$ is \emph{$k$-bounded} if all its runs induce 
$k$-linear substitutions. It is \emph{bounded} if it is $k$-bounded for some 
$k$.


The following theorem gives the expressiveness equivalence of the models we consider.
We do not give the definitions of \MSO graph transductions as our results will
only involve state-based transducers (see~\cite{Filiot-ICLA15} for more details).

\begin{theorem}[\cite{EH01,AC10,AFT12}]\label{Thm:Equiv}
Let $f:A^*\to B^*$ be a function over words.
Then the following conditions are equivalent:
\begin{itemize}
\item $f$ is realized by an $\MSO$ graph transduction,
\item $f$ is realized by a \TWDFT,
\item $f$ is realized by a copyless \sst,
\item $f$ is realized by a bounded \sst.
\end{itemize}
\end{theorem}

\subsection{Transition monoid of transducers}

A (finite) monoid $M$ is a (finite) set equipped with an associative internal law $\cdot_M$ having a neutral element for this law.
A morphism $\eta: M\to N$ between monoids is an application from $M$ to $N$ that preserves the internal laws, meaning that for all $x$ and $y$ in $M$, $\eta(x\cdot_M y)=\eta(x)\cdot_N\eta(y)$.
When the context is clear, we will write $xy$ instead of $x\cdot_M y$.
A monoid $M$ divides a monoid $N$ if there exists an onto morphism  from a submonoid of $N$ to $M$.
A monoid $M$ is said to be
 \emph{aperiodic} if there exists a least integer $n$, called the \emph{aperiodicity index} of $M$, such that for all elements $x$ of $M$, we have $x^n=x^{n+1}$. 

Given an alphabet $A$, the set of words $A^*$ is a monoid equipped with the concatenation law, having the empty word as neutral element. It is called the \emph{free monoid} on $A$.
A finite monoid $M$ \emph{recognizes} a language $L$ of $A^*$ if there exists an onto morphism $\eta:A^*\to M$ such that $L=\eta^{-1}(\eta(L))$.
It is well-known that the languages recognized by finite monoids are exactly the regular languages.

The monoid we construct from a machine is called its \emph{transition monoid}.
We are interested here in aperiodic machines, in the sense that a machine is aperiodic if its transition monoid is aperiodic.
 We now give the definition of the transition monoid for a \TWDFT and an \sst.

\begin{wrapfigure}[6]{R}{0.2\textwidth}
  \begin{tikzpicture}[scale=0.8]
  \draw[|-|] (0,1) -- (2.3,1);
  \node (u) at (1,1.2) {$u$};
  \node (q) at (0,0.6) {$p$};
  \node (p) at (-0.4,0) {$q$};
  \draw plot [smooth] coordinates { (0.1,0.6) (2,0.6) (1,0.4) (1.4,0.2) (0,0)};
  \draw[->] (0,0) -- (p);
  \end{tikzpicture}
  \label{fig:ll-path}
\end{wrapfigure}
\noindent\textbf{Deterministic Two-Way Finite State Transducers} As in the case of automata, the transition monoid of a \TWDFT $T$ is the set of all possible behaviors of $T$ on a word. 
The following definition comes from~\cite{CD15}, using ideas from~\cite{Shepherdson59} amongst others.
As a word can be read in both ways, the possible runs are split
into four relations over the set of states $Q$ of $T$.
Given an input word $u$, we define the left-to-left behavior $\bh_{\ell\ell}(u)$ as the set of pairs $(p,q)$ of states 
of $T$ such that there exists a run over $u$ starting on the first letter of $u$ in state $p$ and exiting $u$ on the left in state $q$ (see Figure on the right).
We define in an analogous fashion the left-to-right, right-to-left and right-to-right behaviors denoted 
respectively $\bh_{\ell r}(u)$, $\bh_{r\ell}(u)$ and $\bh_{rr}(u)$.
Then the transition monoid of a \TWDFT is defined as follows:

  Let $T=(Q,A,\delta,q_0,F)$ be a \TWDFT.  The
  \emph{transition monoid} of~$T$ is $A^*/\!\!\sim_{T}$ where
  $\sim_T$ is the conjunction of the four relations $\LL$,
  $\LR$, $\RL$ and $\RR$ defined for any words $u$, $u'$ of $A^*$ as follows:
  $u \sim_{xy} u'$ iff $\bh_{xy}(u)=\bh_{xy}(u')$, for $x,y\in\{\ell,r\}$.
  The neutral element of this monoid is the class of the empty
  word~$\epsilon$, whose behaviors $bh_{xy}(\epsilon)$ is the identity
  function if $x\neq y$, and is the empty relation otherwise.

Note that since the set of states of $T$ is finite, each behavior relation is of finite 
index and consequently the transition monoid of $T$ is also finite.
Let us also remark that the transition monoid of $T$ does not depend on the output 
and is in fact the transition monoid of the underlying \TWDFA.

\noindent\textbf{Streaming String Transducers} A notion of transition monoid for \sst was defined in~\cite{FKT14}. We give here its formal definition and refer to~\cite{FKT14} for advanced considerations.
In order to describe the behaviors of an \sst, this monoid describes the possible flows of variables
along a run.
Since we give later an alternative definition of transition monoid for \sst, we will call 
it the \emph{flow transition monoid} (\ftm).

Let $T$ be an \sst with states $Q$ and variables $\cX$.
The \emph{flow transition monoid} $M_T$ of $T$ is a set of square matrices over the integers enriched 
with a new absorbent element $\bot$. 
The matrices are indexed by elements of $Q\times\cX$.
Given an input word $u$, the image of $u$ in $M_T$ is the matrix $m$ such that for all states $p,q$ and all variables $X,Y$, $m[p,X][q,Y]=n\in \mathbb{N}$ (resp. $m[p,X][q,Y]=\bot$) if, and only if, 
there exists a run $r$ of $T$ over $u$ from state $p$ to state $q$, and $X$ occurs $n$ times in $\sigma_r(Y)$
(resp. iff there is no run of $T$ over $u$ from state $p$ to state $q$). 

Note that if $T$ is $k$-bounded, then for all word $w$, all the coefficients of its image in $M_T$ are bounded by $k$. The converse also holds.
Then $M_T$ is finite if, and only if, $T$ is $k$-bounded, for some $k$.

It can be checked that the machines given in Example~\ref{example} are aperiodic.
Theorem~\ref{Thm:Equiv} extends to aperiodic subclasses and to first-order logic, as in the case of regular languages~\cite{Schutzenberger65,MP71}.
These results as well as our contributions to these models are summed up in Figure~\ref{fig:overview}.

\begin{theorem}[\cite{FKT14,CD15}]
Let $f:A^*\to B^*$ be a function over words.
Then the following conditions are equivalent:
\begin{itemize}
\item $f$ is realized by a $\FO$ graph transduction,
\item $f$ is realized by an aperiodic \TWDFT,
\item $f$ is realized by an aperiodic $1$-bounded \sst.
\end{itemize}
\end{theorem}

\section{Substitution Transition Monoid}\label{Section:Order}

In this section, we give an alternative take on the definition of the transition monoid of an \sst, and show that both notions coincide on aperiodicity and boundedness.
The intuition for this monoid, that we call the \emph{substitution transition monoid}, is for the elements to take into account not only the multiplicity of the output of each variable in a given run, but also the order in which they appear 
in the output.
It can be seen as an enrichment of the classic view of transition monoids as the set of functions over states equipped with the law of composition.
Given a substitution $\sigma \in \mathcal{S}_{\Vars, B}$, let us denote 
$\tsig$ the projection of $\sigma$ on the set $\cX$, 
i.e. we forget the parts from $B$.
The substitutions $\tsig$ 
are homomorphisms of $\Vars^*$
which form an (infinite) monoid. 
Note that in the case of a
1-bounded \sst, each variable occurs at most once in $\tsig(Y)$.
 

\paragraph{Substitution Transition Monoid of an \sst.}
Let $T$ be an \sst with states $Q$ and variables $\cX$.
The \emph{substitution transition monoid} (\stm) of $T$, denoted $M^\sigma_T$, is a set of partial functions 
$f:Q \rightharpoonup Q\times \mathcal{S}_{\Vars,\emptyset}$. 
%
Given an input word $u$, the image of $u$ in $M^\sigma_T$ is the
 function $f_u$ such that for all states $p$, $f_u(p)=(q,\tsig_r)$ if, and only if,
 there exists a run $r$ of $T$ over $u$ from state $p$ to state $q$ that induces
 the substitution $\tsig_r$.
This set forms a monoid when equipped with the following composition law:
Given two functions $f_u,f_v \in M_T^\sigma$, the function $f_{uv}$
is defined by $f_{uv}(q)=(q'',\tsig\circ\tsig')$ whenever 
$f_u(q)=(q',\tsig)$ and $f_v(q')=(q'',\tsig')$.

We now make a few remarks about this monoid.
Let us first observe that the \ftm of $T$ can  be recovered from its \stm.
Indeed, the matrix $m$ associated with a word $u$ in $M_T$  is easily deduced
from the function $f_u$ in $M_T^\sigma$.
This observation induces an onto morphism from $M_T^\sigma$
to $M_T$, and consequently the $\ftm$ of an \sst divides its $\stm$. This proves that if
the \stm is aperiodic, then so is the \ftm since aperiodicity is preserved by division of monoids.
Similarly, copyless and $k$-bounded \sst (given $k\in\mathbb{N}_{>0}$) are characterized by means of their \stm.
This transition monoid can be separated into two main components:
the first one being the transition monoid of the underlying deterministic one-way automaton, 
which can be seen as a set of functions $Q\to Q$,
while the second one is the monoid $\calS_\cX$ of homomorphisms on $\cX$,
equipped with the composition.
The aware reader could notice that the \stm can be written as the wreath product of the transformation semigroup
 $(\cX^*,\calS_\cX)$ by $(Q,Q^Q)$.
However, as the monoid of substitution is obtained through the closure under composition of 
the homomorphisms of a given \sst, 
it may be infinite.


The next theorem proves that aperiodicity for both notions coincide, since the converse comes from the division of $\stm$ by $\ftm$.
\begin{restatable}{theorem}{PropOrder}\label{Prop:Order}
Let $T$ be a $k$-bounded SST with $\ell$ variables.
If its \ftm is aperiodic with aperiodicity index $n$ then 
its \stm is aperiodic with aperiodicity index at most $n+ (k+1)\ell$.
\end{restatable}

\section{From $1$-bounded \sst to \TWDFT}

The existing transformation of a $1$-bounded (or even copyless) \sst into
an equivalent \TWDFT goes through \MSO transductions, yielding a non-elementary
complexity. We present here an original construction whose complexity is elementary.

\begin{restatable}{theorem}{SSTtoTw}\label{SSTto2w}
Let $T$ be a 1-bounded SST with $n$ states and $m$ variables.
Then we can effectively construct a deterministic 2-way transducer that realizes the same function.
If $T$ is 1-bounded (resp. copyless), then the \TWDFT has $O(m 2^{m2^m} n^n)$ states
(resp. $O(m n^n)$).
\end{restatable}

\begin{proof}
We define the \TWDFT as the composition of a left-to-right sequential transducer, a right-to-left sequential transducer and a 2-way transducer.
Remark that this proves the result as two-way transducers are closed under composition with sequential ones~\cite{CV77}.

The left-to-right sequential transducer does a single pass on the input word and outputs the same word enriched with the transition used by the \sst in the previous step.
The right-to-left transducer uses this information to enrich each position of the input word with the set of useful variables, i.e  the variables 
that flow to an output variable according to the partial run on the suffix read.
The two sequential transducers are quite standard. They realize length-preserving functions that simply enrich the input word with new information. 
The last transducer is more interesting: it uses the enriched information to follow the output structure of $T$.
The \emph{output structure} of a run is a labeled and directed graph such that, for each variable $X$ useful at a position $j$, we have two nodes $X_i^j$ and $X_o^j$ linked by a path whose concatenated labels form the value stored in $X$ at position $j$ of the run (see~\cite{FKT14} and Figure~\ref{Fig:OutputStructure}).


\begin{figure}[t!]
\begin{center}
\begin{tikzpicture}[yscale=0.9]
\node[draw,shape=circle,minimum size=0.8cm,inner sep=0pt,outer sep=0pt] (q0) at (0.5,6.5) {$q_{i-1}$};
\node[draw,shape=circle,minimum size=0.8cm,inner sep=0pt,outer sep=0pt] (q1) at (3.5,6.5) {$q_i$};
\path[->,>=latex] (q0) edge [above] node {$a_i$} (q1);

\node[draw,shape=circle,minimum size=0.8cm,inner sep=0pt,outer sep=0pt] (q2) at (6.5,6.5) {$q_{i+1}$};
\path[->,>=latex] (q1) edge [above] node {$a_{i+1}$} (q2);

\node[draw,shape=circle,minimum size=0.8cm,inner sep=0pt,outer sep=0pt] (q3) at (9.5,6.5) {$q_{i+2}$};
\path[->,>=latex] (q2) edge [above] node {$a_{i+2}$} (q3);

\node[align=center,shape=rectangle,draw,rounded corners=1mm] (s1) at (2,5.8) { X:= aa \\  \\ Z:=aZab};
\node[align=center,shape=rectangle,draw,rounded corners=1mm] (s2) at (5,5.8) 
	{X:= aXc \\ Y:= bZ \\  };
\node[align=center,shape=rectangle,draw,rounded corners=1mm] (s3) at (8,5.8) { X:= XaaYb \\  \\ };

\node (xi) at (0,4.5) {$X_i$:};
\draw[dotted] (xi) -- (9,4.5);
\node (x0) at (0,4) {$X_o$:};
\draw[dotted] (x0) -- (9,4);

\node (yi) at (0,3.5) {$Y_i$:};
\draw[dotted] (yi) -- (9,3.5);
\node (y0) at (0,3) {$Y_o$:};
\draw[dotted] (y0) -- (9,3);

\node (zi) at (0,2.5) {$Z_i$:};
\draw[dotted] (zi) -- (9,2.5);
\node (z0) at (0,2) {$Z_o$:};
\draw[dotted] (z0) -- (9,2);

\foreach \x/\y/\n in 
	{2/4.5/1,5/4.5/2,8/4.5/3, 5/4/4, 8/4/5,
	5/3.5/6, 8/3/7,
	2/2.5/8, 2/2/9, 5/2/10}
	\node[minimum size=0cm,inner sep=0pt,outer sep=0pt] (n\n) at (\x,\y) {$\bullet$};
	
\path[->,>=latex] (n3) edge [above] node {$\epsilon$} (n2);
\path[->,>=latex] (n2) edge [above] node {$a$} (n1);
\path[->,>=latex] (n1) edge [left, below] node {$aa$} (n4);
\path[->,>=latex] (n4) edge [above] node {$c$} (n5);
\path[->,>=latex] (n5) edge [below] node {$aa$} (n6);
\path[->,>=latex] (n6) edge [above] node {$b$} (n8);

\path[->,>=latex,dashed] (n8) edge [above] node {$a$} (0.5,2.5);
\path[->,>=latex,dashed] (0.5,2) edge [above] node {} (n9);
\path[->,>=latex] (n9) edge [above] node {$ab$} (n10);
\path[->,>=latex] (n10) edge [above] node {$\epsilon$} (n7);
\path[->,>=latex,dashed] (n7) edge [above] node {$b$} (9.3,3.8);

\end{tikzpicture}
\caption{The output structure of a partial run of an \sst used in the proof of Theorem~\ref{SSTto2w}.}\label{Fig:OutputStructure}
\end{center}
\end{figure}
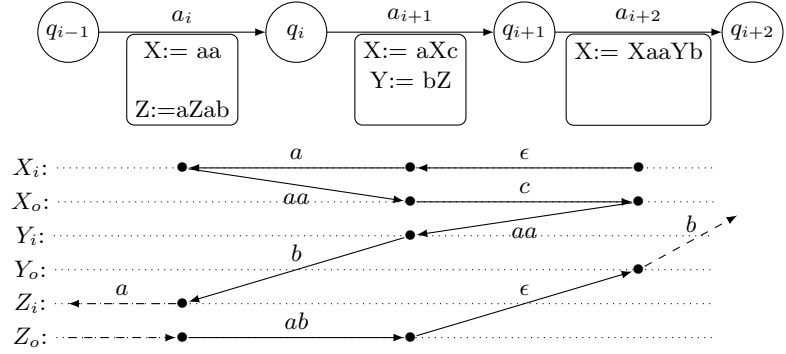
The transition function of the two-way transducer is described in Figure~\ref{Fig:Update}.
It first reaches the end of the word and picks the first variable to output.
It then rewinds the run using the information stored by the first sequential transducer, producing the said variable using the local update function.
When it has finished to compute and produce 
a variable $X$, it switches to the following one using the information of the second transducer to know which variable $Y$ $X$ is flowing to, and starts producing it.
Note that such a $Y$ is unique thanks to the  $1$-boundedness property.
If $T$ is copyless, then this information is local and the second transducer can be bypassed.

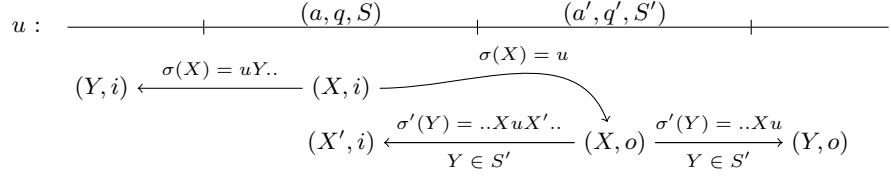
\begin{figure}[t!]
\begin{center}
\begin{tikzpicture}[scale=0.9]
\draw[-|] (0,5.2) -- (2,5.2);
\draw[-|] (2,5.2) -- (6,5.2);
\draw[-|] (6,5.2) -- (10,5.2);
\draw (10,5.2) -- (12,5.2);

\node (u) at (-0.6,5.2) {$u$ :};
\node (qas) at (4,5.4) {$(a,q,S)$};
\node (qas2) at ( 8,5.4) {$(a',q',S')$};

\node (xi) at (4,4.3) {$(X,i)$};
\node (yi) at (0.5,4.3) {$(Y,i)$};
\path[->] (xi) edge [above] node {\scriptsize{$\sigma(X)=uY..$}} (yi);

\node (xo) at (8,3.5) {$(X,o)$};
\path[->] (xi) edge [in=110,out=0,above] node {\ \ \scriptsize{$\sigma(X)=u$}} (xo);

\node (yo) at (11,3.5) {$(Y,o)$};
\path[->] (xo) edge [above] node {\scriptsize{$\sigma'(Y)=..Xu$}} (yo);
\path[] (xo) edge [below] node {\scriptsize{$Y\in S'$}} (yo);

\node (xpi) at (4, 3.5) {$(X',i)$};

\path[->] (xo) edge [above] node {\scriptsize{$\sigma'(Y)=..XuX'..$}} (xpi);
\path[] (xo) edge [below] node {\scriptsize{$Y\in S'$}} (xpi);

\end{tikzpicture}
\end{center}

\caption{The third transducer follows the output structure. 
States indexed by $i$ correspond to the beginning of a variable, while states indexed by $o$ correspond the end.
$\sigma$ (resp. $\sigma'$) stand for the substitution at position $a$ (resp. $a'$).}\label{Fig:Update}
\end{figure}


Regarding complexity, a careful analysis of the composition of a one-way transducer of size 
$n$ with a two-way transducer of size $m$  from~\cite{PhDartois,CD15} shows that this can 
be done by a two-way transducer of size $O(m n^n)$.
Then given a $1$-bounded SST with $n$ states and $m$ variables, we can construct a deterministic two-way transducer of size $O(m 2^{m2^m} n^n)$. 
If $T$ is copyless, the second sequential transducer is omitted, resulting in a size of $O(m n^n)$.
\end{proof}

\begin{restatable}{theorem}{SSTtoTwAp}\label{SSTto2w:Ap}
Let $T$ be an aperiodic 1-bounded SST. Then the equivalent \TWDFT 
constructed using Theorem~\ref{SSTto2w} is also aperiodic.
\end{restatable}

\begin{proof}
The aperiodicity of the three transducers gives the result as aperiodicity is preserved by composition of a one-way by a two-way~\cite{CD15}.
The aperiodicity of the two sequential transducers is straightforward since their runs depend respectively on the underlying automaton and the update function.
The aperiodicity of the \TWDFT comes from the fact that since it follows the output structure of the \sst, its partial runs are induced by the flow of variables and their order in the substitutions, which is an information contained in the \ftm and thus aperiodic thanks to Theorem~\ref{Prop:Order}.
\end{proof}

\section{From \TWDFT to copyless SST}

In~\cite{AC10}, the authors give a procedure to construct a copyless \sst from a \TWDFT.
This procedure uses the intermediate model of heap based transducers.
We give here a direct construction with similar complexity. This simplified presentation allows us tu 
prove that the construction preserves the aperiodicity.

\begin{restatable}{theorem}{TWtoSSTb}\label{Thm:2wSSTb}
Let $T$ be a \TWDFT with $n$ states.
Then we can effectively construct a copyless SST with $O((2n)^{2n})$ states and $2n-1$ variables that computes the same function.
\end{restatable}

\begin{proof} (Sketch of)
The main idea is for the constructed \sst to keep track of the right-to-right behavior of the prefix read until the current position, similarly to the construction of Shepherdson~\cite{Shepherdson59}. This information can be updated upon reading a new letter, constructing a one-way machine recognizing the same input language. 
The idea from~\cite{AFT12} is to have one variable per possible right-to-right run, which is bounded by the number of states.
However, since two right-to-right runs from different starting states can merge, this construction results in a 1-bounded \sst.
To obtain copylessness, we keep track of these merges and the order in which they appear. Different variables are used to store the production of each run before the merge, and one more variable stores the production after.

The states of the copyless \sst are represented by sets of labeled trees having the states of the input \TWDFT as leaves.
Each inner vertex represents one merging, and two leaves have a common ancestor if the right-to-right runs from the corresponding states merge at some point.
Each tree then models a set of right-to-right runs that all end in a same state. Note that it is necessary to also store the end state of these runs.
For each vertex, we use one variable to store the production of the partial run corresponding to the outgoing edge. 

Given such a state and an input letter, the transition function can be defined by adding to the set of trees the local transitions at the given letter, and then reducing the resulting graph in a proper way (see Figure~\ref{fig-2WcpSST}).

Finally, as merges occur upon two disjoint sets of states of the \TWDFT (initially singletons), the number of merges, and consequently the number of inner vertices of our states, is bounded by $n-1$.
Therefore, an input \TWDFT with $n$ states can be realized by an \sst having $2n-1$ variables.
Finally, as states are labeled graphs, Cayley's formula yields an exponential bound on the number of states.
\end{proof}
%
%
\begin{figure}[t!]

\begin{center}
\begin{minipage}[l]{.32\linewidth}

\vspace{-13pt}

\begin{tikzpicture}[scale=0.7]
\draw[|-|] (1,6) -- (5,6);
\draw[-|] (5,6) -- (6,6);
\draw[dotted] (1,6) -- (1,0);
\draw[dotted] (5,6) -- (5,0);
\draw[dotted] (6,6) -- (6,0);
\draw[dotted] (4,6) -- (4,0);

\node (a) at (5.5,6.3) {$a$};
\node (u) at (3,6.3) {$u$};

\foreach \i in 
	{5,4,3,2,1,0}
	{\node (q\i) at (4.5,5.5 - \i) {$q_\i$};
	\node[color=red] (a\i) at (5.5,5.5 - \i)  {$q_\i$};}

\node (q01) at (3,5) {$\bullet$};
\node (q012) at (2,4.5) {$\bullet$};
\node (q45) at (3,1.5) {$\bullet$};

\draw[->] (q0) -- (q01);
\draw[->] (q1) -- (q01);
\draw[->] (q2) -- (q012);
\draw[->] (q01) -- (q012);
\draw[->] (q4) -- (q45);
\draw[->] (q5) .. controls (1,1) and (1,1).. (q45);
\draw[->] (q3) -- (a2);

\draw[->] (q012) .. controls (1,4) and (1,4).. (q3);
\draw[->] (q45) .. controls (3.5,2.2) and (3.5,2.2).. (q3);

\draw[->,color=red] (a0) -- (6.5,5.5);
\draw[->,color=red] (a1) -- (q0);
\draw[->,color=red] (a2) -- (6.5,2.5);
\draw[->,color=red] (a3) -- (q2);
\draw[->,color=red] (a4) -- (q4);
\draw[->,color=red] (a5) -- (6.5,1.5);
\end{tikzpicture}
\end{minipage}
\begin{minipage}[l]{.05\linewidth}
$\Rightarrow$
\end{minipage}
\begin{minipage}[r]{.25\linewidth}

\vspace{-13pt}
\begin{tikzpicture}[scale=0.7]
\draw[|-|] (3,6) -- (6,6);
\draw[dotted] (3,6) -- (3,0);
\draw[dotted] (6,6) -- (6,0);
\node (u) at (4.5,6.3) {$ua$};
\foreach \i in 
	{5,4,3,2,1,0}
	{
	\node (q\i) at (5.5,5.5 - \i)  {$q_\i$};}

\node (q13) at (4.5,3.5) {$\bullet$};
\node (q134) at (3.5,2.5) {$\bullet$};

\draw[->] (q1) -- (q13);
\draw[->] (q3) -- (q13);
\draw[->] (q13) -- (q134);
\draw[->] (q4) -- (q134);

\draw[->] (q0) -- (6.5,5.5);
\draw[->] (q2) -- (6.5,2.5);
\draw[->] (q5) -- (6.5,1.5);

\draw[->] (q134) ..controls (4.5,2.6) and (4.5,2.6) .. (q2);

\end{tikzpicture}
\end{minipage}
\begin{minipage}[r]{.35\linewidth}
\caption{Left: The state of the \sst is represented in black. The red part corresponds to the local transitions of the \TWDFT. Right: After reading $a$, we reduce the new forest by eliminating the useless branches and shortening the unlabeled linear paths.}\label{fig-2WcpSST}
\end{minipage}\hfill
\end{center}
\end{figure}
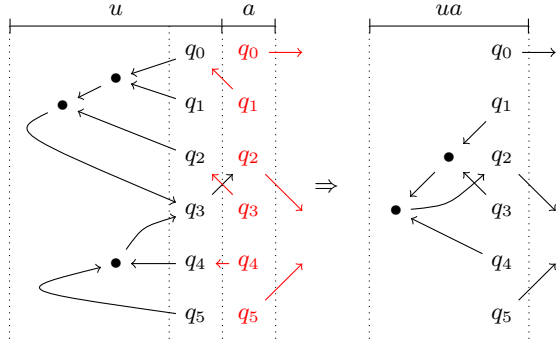

Moreover, this construction preserves aperiodicity:
\begin{restatable}{theorem}{TWtoSSTApb}
Let $T$ be an aperiodic \TWDFT.
Then the equivalent \sst constructed using Theorem~\ref{Thm:2wSSTb} is also aperiodic.
\end{restatable}
\begin{proof}
If the input \TWDFT is aperiodic of index $n$, then
for any word $w$, $w^n$ and $w^{n+1}$ merge the same partial runs for the four kinds of behaviors, by definition, and in fact the merges appear in the same order.
As explained earlier, the state $q_1$ (resp. $q_2$) reached by the constructed \sst over the inputs $uw^{n}$  (resp. $uw^{n+1}$) represents the merges of the right-to-right runs of $T$ over $uw^{n}$ (resp. $uw^{n+1}$).
Since these runs can be decomposed in right-to right runs over $u$ and partial runs over $w^n$ and $w^{n+1}$, the merge equivalence between $w^n$ and $w^{n+1}$ implies that $q_1 = q_2$. 
Moreover, since variables are linked to these merges, the aperiodicity of the merge equivalence implies the aperiodicity of both the underlying automaton and the substitution function of the \sst, concluding the proof.

\end{proof}

As a corollary, we obtain that the class of aperiodic copyless \sst 
 is expressively equivalent to first-order
definable string-to-string transductions.
\begin{corollary}
Let $f:A^*\to B^*$ be a function over words. Then  $f$ is realized by a $\FO$ graph transduction
iff it is realized by an aperiodic copyless \sst.
\end{corollary}

\section{From $k$-bounded to $1$-bounded \sst}


The existing construction from $k$-bounded to $1$-bounded, presented in~\cite{AFT12},
builds a copyless \sst. 
We present an alternative construction that, given a $k$-bounded \sst, directly builds an
equivalent $1$-bounded \sst. We will prove that this construction preserves aperiodicity.


\begin{restatable}{theorem}{KtoO}\label{thm:kto1}
Given a $k$-bounded \sst $T$ with $n$ states and $m$ variables, we can effectively construct 
an equivalent $1$-bounded \sst. This new \sst has $n2^N$ states and $mkN$ variables,
where $N=O(n^n (k+1)^{nm^2})$ is the size of the flow transition monoid $M_T$.
\end{restatable}

\begin{proof}
In order to move from a $k$-bounded \sst to a $1$-bounded \sst, the natural idea is to 
use copies of each variable.
However, we cannot maintain $k$ copies of each variable all the time: suppose that $X$ flows into $Y$
and $Z$, which both occur in the final output. If we have $k$ copies of $X$, we cannot produce 
in a $1$-bounded way (and we do not need to) $k$ copies of $Y$ and $k$ copies of $Z$. 

Now, if we have access to a look-ahead information, we can guess how many copies of each variable are needed, and we can easily construct a copyless \sst by using exactly the right number of copies for each variable and at each step.
The construction relies on this observation. We simulate a look-ahead through a subset construction, having copies of each variable for each possible behavior of the suffix.
Then given a variable and the behavior of a suffix, we can maintain the exact number of variables needed and perform a copyless substitution to a potential suffix for the next step.
However, since the \sst is not necessarily co-deterministic, a given suffix can have multiple successors, and the result is that its variables flow to variables of different suffixes. As variables of different suffixes are never recombined, we obtain a $1$-bounded \sst.
\end{proof}

\vspace{-.1cm}
\begin{restatable}{theorem}{KtoOAp}
Let $T$ be an aperiodic $k$-bounded \sst.
Then the equivalent $1$-bounded \sst constructed using 
Theorem~\ref{thm:kto1} is also aperiodic.
\end{restatable}
As a corollary, we obtain that for 
the class of aperiodic bounded \sst is expressively equivalent to first-order
definable string-to-string transductions.
\begin{corollary}
Let $f:A^*\to B^*$ be a function over words. Then $f$ is realized by a $\FO$ graph transduction
iff it is realized by an aperiodic bounded \sst 
($k\in\mathbb{N}_{>0}$).
\end{corollary}

\section{Perspectives}
%
%
There is still one model equivalent to the generic machines whose aperiodic subclass elude our scope yet,
namely the \emph{functional two-way transducers}, which correspond to non-deterministic two-way transducers realizing a function.
To complete the picture, a natural approach would then be to consider the constructions from~\cite{dS13} and prove that aperiodicity is preserved.
One could also think of applying this approach to other varieties of monoids, such as the $\mathcal{J}$-trivial monoids, equivalent to the boolean closure of existential first-order formulas $\mathcal{B}\Sigma_1[<]$.
Unfortunately, the closure of such transducers under composition requires some strong properties on varieties (at least closure under semidirect product) which are not satisfied by varieties less expressive than the aperiodic.
Consequently the construction from \sst to \TWDFT cannot be applied.
On the other hand, the other construction could apply, providing one inclusion.
Then an interesting question would be to know where the corresponding fragment of logic would position.

\bibliography{biblio}
\bibliographystyle{abbrv}

\newpage
\section*{Appendix}
\appendix

\subsection*{Substitution Transition Monoid}
\PropOrder*
\begin{proof}
Let $T$ be a $k$-bounded SST.
We define a \emph{loop} as the run induced by a pair $(q,u)\in Q\times A^*$ such that $\delta(q,u)=q$.
Suppose now that $M_T$ is aperiodic, and let $n$ be its aperiodicity index.
Wlog, we assume that the transition function of $T$ is complete.
This implies that for all states $p$ of $T$, there exists a state $q$ such that $p\xrightarrow{u^{n}} q \xrightarrow{u} q$.
Then if the image in \stm of the loops (i.e. the set of all $\tsig$ such that there exists a loop $(q,u)$ such that $f_u(q)=(q,\tsig)$) are aperiodic with index $m$, then the \stm is aperiodic with index at most $n+m$.

Consequently, in the following $\sigma$ denotes the substitution of a loop of $T$, and we aim to prove that
 $\tsig^{(k+1)\ell}=\tsig^{(k+1)\ell+1}$.

Before proving this though, we define the relation $\lessdot\subseteq \cX\times\cX$ as follows.
Given two variables $X$ and $Y$, we have $X\lessdot Y$ if there exists a positive integer $i$ such that $X$ flows into $Y$ in $\sigma^i$. This relation is clearly transitive. The next lemma proves that it is also anti-symmetric, hence we can use this relation as an induction order to prove the result.
\begin{lemma}\label{Lemma:flow}
Given two different variables $X$ and $Y$, if $X\lessdot Y$, then $Y\not\!\!\lessdot X$.

\end{lemma}
\begin{proof}
We proceed by contradiction.
Assume that there exist two different variables $X$ and $Y$ and two integers $i$ and $j$ such that $X$ occurs in $\sigma^i(Y)$ and $Y$ occurs in $\sigma^j(X)$.

Then for any $k>0$, $X$ occurs in $\sigma^{k(i+j)}(X)$
and $Y$ occurs in $\sigma^{k(i+j)+j}(X)$.
As $T$ is aperiodic of index $n$, for $k$ large enough it means that both $X$ and $Y$ occur in both $\sigma^{n}(X)$ and $\sigma^{n}(Y)$. 
Then $\sigma^{2n}(X)$ contains both $\sigma^{n}(X)$ and $\sigma^{n}(Y)$ and thus contains 
at least two occurrences of $X$ and $Y$.
Then by aperiodicity we have $\sigma^{2n}(X)=\sigma^{n}(X)$ thus $\sigma^n(X)$ contains two occurrences of $X$.
By iterating this process, we prove that the number of occurrences of $X$ in  $\sigma^{n}(X)$ is not bounded, 
yielding a contradiction.
\end{proof}
We now prove that for all variables $X$ in $\cX$, $\tsig^{(k+1)\ell}(X)=\tsig^{(k+1)\ell+1}(X)$ by treating the following two cases:
\begin{itemize}
\item If $X\in \sigma(X)$, then either $\tsig(X)=X$ and then $\tsig^2(X)=\tsig(X)$, or there exists $Y\neq X$ such that $Y\in\tsig(X)$. In the latter case, we get by iteration that for all $i>0$, $|\tsig^i(X)|> \Sigma_{j<i} |\tsig^j(Y)|$. Then as $T$ is $k$-bounded, we have $|\tsig^i(X)|\leq k\ell$ and thus  $\Sigma_{j<i} |\tsig^j(Y)|$ is bounded, and $\tsig^{k\ell}(Y)=\epsilon$, which proves that $\tsig^{k\ell}(X)=\tsig^{k\ell+1}(X)$.
\item If $X\not\in\sigma(X)$, let us consider the relation $\lessdot$. By Lemma~\ref{Lemma:flow} this relation is cycle-free.
Then there is a lesser level, on which there are the variables $Y$ such that $\tsig(Y)\subseteq \{Y$\}. There, either $\tsig(Y)=\emptyset$ and aperiodicity becomes trivial, or $\tsig(Y)=Y$ and the case was dealt with in the previous point and is thus aperiodic with index $k\ell$.
Now we can end the proof by reasoning by induction on $\lessdot$, as if $X\not\in\sigma(X)$ and all variables $Y\lessdot X$ are aperiodic with index $i$, then $\tsig^{i+1}(X)$ can be written as the concatenation of 
$\tsig^{i}(Y)$, for aperiodic variables $Y$ of index $i$. Then $\tsig^{i+1}(X)=\tsig^{i+2}(X)$.
The proof is concluded by noticing that the length of the longest chain of $\lessdot$ is bounded by $\ell$.\end{itemize}\end{proof}

\subsection*{From $1$-bounded \sst to \TWDFT}

\SSTtoTw*

\begin{proof}
Let $T=(A,B,Q,q_0,Q_f,\delta,\cX,\rho,F)$ be a SST.
Let us construct a two-way transducer $\cA$ that realizes the same function.
The transducer $\cA$ will follow the output structure (see Figure~\ref{Fig:OutputStructure}) of $T$ and construct the output as it appears in the structure.
To make the proof easier to read, we define $\cA$ as the composition of a left-to-right sequential transducer $\cA'$, a right-to-left sequential transducer $\cA''$ and a 2-way transducer $\cB$.
Remark that this proves the result as two-way transducers are closed under composition with sequential ones.
The transducer $\cA'$ does a single pass on the input and enriches it with the transition used by $T$ in the previous step.
The second transducer uses this information, and enriches the input word with the set of variables corresponding to the variables that will be produced from this position.
The last transducer is more interesting: it uses the enriched information to follow the output structure of $T$. 
The \emph{output structure} of a run is a labeled and directed graph such that, for each variable $X$ useful at a position $j$, we have two nodes $X_i^j$ and $X_o^j$ linked by a path whose concatenated labels form the value stored in $X$ at position $j$ of the run (see~\cite{FKT14} and Figure~\ref{Fig:OutputStructure}).
The set of variables will be used to clear the non determinism due to the $1$-bounded property. Note that in the case of a copyless SST, the transducer $\cA''$ can be omitted.
We now explain how the several transducers behave on a given run $r=q_0\xrightarrow{a_1} q_1 \ldots \xrightarrow{a_n} q_n$.

The transducer $\cA'=(A,A\times Q,Q\times \uplus\{f\},q_0,\alpha,\beta,\{f\})$, which enriches the input word with the transitions of the previous step, can be done easily with a 1-way transducer, which first stores the transition taken in its state, then outputs it along the current letter read.
Given a letter $a$ and a state $q$, we have the transitions $(q,a) \xrightarrow{a\mid (a,q)} \delta(q,a)$.
We also have $(q,\dashv)\xrightarrow{\dashv\mid (\dashv,q)} f$.
Then on the run $r$, if $u=a_1\ldots a_n$, we get the output word $\cA'(\vdash u\dashv )=(a_1,q_0)(a_2,q_1)\ldots (a_n,q_{n-1})(\dashv,q_n)$. 

The transducer $\cA''=(A\times Q,A\times Q \times 2^\mathcal{X},2^\mathcal{X}\uplus\{i,f'\},i,\alpha,\beta,\{f'\})$, which enriches each letter of the input word with the variables effectively produced from this step, can be done easily with a right-to-left sequential transducer, which starts with the variables appearing in $F(q_n)$.
Given a letter $(a,q)$ and a set $S$, we have $S \xrightarrow{(a,q)\mid(a,q,S)} S'$ where $S'=\{X\in\cX \mid \exists Y \in S\ X\text{ occurs in } \rho(q,a,Y)\}$.
Then given an input word $(a_1,q_0)(a_2,q_1)\ldots (\dashv,q_n)$, we define $S_n=F(q_n)$, and $S_{i-1}=\{X \in\mathcal{X}\mid \exists Y\in F(q_n)\  X\in\sigma_{q_{i-1},a_i\ldots a_n}(Y)\}$, we get the output word $\cA''\circ\cA'(\vdash u \dashv)=(a_1,q_0,S_1)(a_2,q_1,S_2)\ldots (a_n,q_{n-1},S_n) (\dashv,q_n,\emptyset)$. 

The aim of the third transducer $\cB$ is to follow the output structure of $T$, which can be defined as follows:
The \emph{output structure} of a run is a labeled and directed graph such that, for each variable $X$ useful at a position $j$, we have two nodes $X_i^j$ and $X_o^j$ linked by a path whose concatenated labels form the value stored in $X$ at position $j$ of the run.
Formally, the output structure of a run $q_0\xrightarrow{u}q_n$ is the oriented graph over $\cX\times [1,|u|]\times \{i,o\}$ whose edges are labeled by output and are of the form:
\begin{itemize}
\item $((X,j,i),v,(Y,j-1,i))$ if $\rho(q_{j-1},a_j,X)$ starts with $vY$.
\item $((X,j,o),v,(Y,j-1,i))$ if there exists $Z$ such that  $XvY$ appears in $\rho(q_{j-1},a_j,Z)$,
\item  $((X,j,o),v,(Y,j+1,o))$ if $\rho(q_{j-1},a_j,Y)$ ends by $Xv$,
\item  $((X,j,i),v,(X,j+1,o))$ if $\rho(q_{j-1},a_j,X)=v$.
\end{itemize}
We furthermore restrict to the connected component corresponding to the actual output of the run.

Now let $\cB=(A\times Q\times 2^\mathcal{X},B,P,p_0,\mu,\nu,\{f\})$ be defined by:
\begin{itemize}
\item $P=\cX\times\{i,o\}\uplus \{p_0,f\}$ is the set of states.
		The transducer does a first left-to-right reading of the input in state $p_0$.
		The subset $\cX\times\{i,o\}$ will then be used to follow the output structure while keeping track of which variable we are currently producing. The set $\{i,o\}$ stands for $in$ and $out$ and corresponds to the similar notions in the output structure. 
		Informally, $in$ states will move to the left, while $out$ states move to the right.
		The states $p_0$ and $f$ are new states that are respectively initial and final.
\end{itemize}
The transition function $\mu:P\times A\times Q\times 2^\mathcal{X}\to P\times\{-1,0,+1\}$ and the production function $\nu: P\times A\times Q\times 2^\mathcal{X}\to B^*$ are detailed below.
In the following, we consider that the transducer is in state $p$ reading the triplet $t=(a,q,S)$ or one of the endmarkers (see Figure~\ref{Fig:Update}).
\begin{itemize}
\item If $p=p_0$ and $a\neq \dashv$, then we set $\mu(p_0,t)=(p_0,+1)$ and $\nu(p_0,t)=\epsilon$.

\item If $p=p_0$ and $a=\dashv$, then if $F(q)$ starts by $uX$ with $u\in B^*$ and $X\in\cX$, then $\mu(p,t)=((X,i),-1)$ and $\nu(p,t)=u$.
\item If $p=(X,i)$, and $t\neq\vdash$ then:
\begin{itemize}
\item either $\rho(q,a)(X)=u\in B^*$ and does not contain any variable, and we set
 $\mu(p,t)=((X,o),+1)$ and $\nu(p,t)=u$,
\item or
  $\rho(q,a)(X)$ starts by $uY$ with $u\in B^*$ and $Y\in\cX$, then $\mu(p,t)=((Y,i),-1)$ and $\nu(p,t)=u$.
\end{itemize}
\item If $p=(X,i)$, and $t=\vdash$ then 
		$\mu(p,t)=((X,o),+1)$ and $\nu(p,t)=\epsilon$.
\item If $p=(X,o)$ and $a\neq\dashv$, then let $Y$ be the unique variable of $S$ such that $X$ appears in $\rho(q,a)(Y)$. Then we have:
\begin{itemize}
	\item either $\rho(q,a)(Y)$ ends by $Xu$ with $u$ in $B^*$ and we set
	$\mu(p,t)=((Y,o),+1)$ and $\nu(p,t)=u$,
	\item or $\rho(q,a)(Y)$ is of the form $(B\cup \cX)^*XuX'(B\cup \cX)^*$ and we set
	$\mu(p,t)=((X',i),-1)$ and $\nu(p,t)=u$.
\end{itemize}
	Note that the unicity of such $Y$ in $S$ is due to the $1$-boundedness property.
	If $T$ is copyless, then this information is irrelevant and $\cA''$ can be bypassed.
	
\item If $p=(X,o)$, $q\in Q_f$ and $a=\dashv$ then:\\
	either $F(q)$ ends by $Xu$ with $u$ in $B^*$ and we set
	$\mu(p,t)=(f,+1)$ and $\nu(p,t)=u$,\\
	or $F(q)$ is of the form $(B\cup \cX)^*XuX'(B\cup \cX)^*$ and we set
	$\mu(p,t)=((X',i),-1)$ and $\nu(p,t)=u$.
\end{itemize}

Then we can conclude the proof as $T=\cB\circ\cA''\circ\cA'$ and 2-way transducers are closed by composition \cite{CV77}.

Regarding complexity, a careful analysis of the composition of a one-way transducer of size 
$m$ with a two-way transducer of size $n$  from~\cite{PhDartois,CD15} shows that this can 
be done by a two-way transducer of size $O(n m^m)$.
Then given a $1$-bounded SST with $n$ states and $m$ variables, we can construct a deterministic two-way transducer of size $O(m (2^m)^{2^m} n^n)=O(m 2^{m2^m} n^n)$.
If $T$ is copyless, the sequential right-to-left transducer can be omitted, and
the resulting \TWDFT is of size $O(m n^n)$.
\end{proof}

\SSTtoTwAp*

\begin{proof}
We prove separately the aperiodicity of the three transducers.
Then the result comes from the fact that aperiodicity is preserved by composition of a one-way by a two-way~\cite{CD15}.

First, consider the transducer $\cA'$. It is a one-way transducer that simply enriches the input word with transitions from $T$, each enrichment corresponding to the transition taken by $T$ in the previous step.
Then since $T$ is aperiodic, so is its underlying automaton.
Then the enrichment and thus $\cA'$ are aperiodic.

Secondly, given an input word,  the transducer $\cA''$ stores at each position the set of variables that will be output by $T$.
Now as $T$ is aperiodic, the flow of variable is aperiodic. Thus the value taken by this set is aperiodic and so is $\cA''$.

Now, consider the transducer $\cB$ and a run $r$ of $\cB$ over $u^n$ starting in state $p$.
Note that the fact that there exists a run over an enriched input word $v$ implies that it is well founded, meaning that it is the image of some word of $A^*$ by $\cA''\circ\cA'$.
If $p$ is of the form $(X,i)$, then the run starts from the right of $u^n$ and follows the 
substitution $\sigma_r(X)$. It exits $u^n$ either in state $(X,o)$ on the right if $\sigma_r(X)$ is a word of $B^*$, or in a state $(Y,i)$ on the left where $Y$ is the first variable appearing in $\sigma_r(X)$.
In both cases the state at the end of the run only depends on the underlying automata of $T$ and the order of variables appearing in the substitution induced by the run. Since the substitution transition monoid is aperiodic if $T$ is aperiodic by Theorem~\ref{Prop:Order}, and then a similar run exists over $u^{n+1}$.

Finally, if $p$ is of the form $(X,o)$, then the state in which the run exits $u^n$ depends on the unique variable $Y$ such that $X$ belongs to $\sigma_r(Y)$ and $Y$ belongs to the set of variables of the last letter of the input.
Then the run follows the substitution $\sigma_r(Y)$.
It will exit the input word in state $(X',i)$ on the left if $XX'$ appears in $\tsig_r(Y)$ for some variable $X'$ and in state $(Y,o)$ otherwise. As the flow of variable as well as the underlying automaton are aperiodic, a similar run exists over $u^{n+1}$.

We conclude the proof by noticing that the same arguments will hold to reduce runs over $u^{n+1}$ to runs over $u^n$.
\end{proof}

\subsection*{From \TWDFT to copyless \sst}
\TWtoSSTb*

Let $T=(A,B,Q,i,\delta,\gamma,F)$ be a \TWDFT.
Let us suppose that the transducer $T$ starts to read its input from the end, and not from the beginning, i.e., given an input $w$, the initial configuration is $(q_0,|w|+1)$.
Moreover, let us suppose that for any transition $(p,\vdash,q,m)$ of $T$, $\gamma(p,\vdash,q,m) = \epsilon$.
Note that any \TWDFT can be transformed with ease into a transducer satisfying those two properties.

In order to reproduce the behavior of $T$ with an \sst $T'$, we need to keep track of the right-to-right runs of $T$.
Moreover, as we want $T'$ to be copyless, it is not possible to store the production of a right-to-right run into a single variable, since two different runs might share a common suffix, and require copy in order to update the corresponding variables.
This leads us to modelize the right-to-right runs of $T$ with rooted forests whose vertices are included into $2^Q \setminus \emptyset$.
The idea is that, given two states $q_1, q_2 \in Q$, a merging between the right-to-right runs starting from $q_1$ and $q_2$ can be represented by adding an edge from both $\{q_1\}$ and $\{q_2\}$ towards $\{q_1,q_2\}$.
Formally, we use the set $\mathcal{F}_Q$ of rooted forests $G = (V,E)$ such that $V$ is a subset of $2^Q \setminus \emptyset$, and the following properties are satisfied.
\begin{itemize}
\item
The roots of $G$ are disjoint subsets of $Q$.
\item
For every vertex $s$, the sons of $s$ are disjoint proper subsets of $s$.
\end{itemize}
Note that two graphs of $\mathcal{F}_Q$ with same set of vertices are equal, hence each element of $\mathcal{F}_Q$ is uniquely defined by its set of vertices.
In order to also keep track of the target states of the right-to-right runs, the states of $T'$ are pairs $(G,\map)$, where $G \in \mathcal{F}_Q$, and $\map$ is an injective function mapping each tree of $G$, corresponding to a set of merging runs, to an element of $Q$, corresponding to the target state of those runs.
In the following definitions, for every vertex $v$ of $G$ we will usually denote by $\map(s)$ the state $\map(T_s)$, where $T_s$ is the tree containing the vertex $s$.

For every word $w$ over the alphabet $\bar{A} = A \cup \{ \vdash, \dashv \}$, we now expose an inductive construction of a pair $(G_w,\map_w)$, where $G_w \in \mathcal{F}_Q$ and $\map_w$ maps the roots of $G_w$ to $Q$, that contains all the information concerning the right-to-right runs of $T$ over the input $w$, and their mergings.
Formally, for every state $q \in Q$ that appears at least in one vertex of $G_w$, let $s_q \in 2^Q$ be the vertex of $G_w$ of minimal size such that $q \in s_q$.
Then we want the following property to be satisfied.

\begin{description}
\item[\prop{1}]
For all $(p,q) \in \bh_{rr}(w)$, we have $\map_{w}(s_p) = q$.
\end{description}

Let $G_{\epsilon} \in \mathcal{F}_Q$ be the graph on $0$ vertex, and let $\map_{\epsilon}: V \rightarrow Q$ be the empty function.
Note that $\prop{1}$ is trivially satisfied for $w = \epsilon$, as $\bh_{rr}(\epsilon)$ is empty, by definition.
Now, let $w \in \bar{A}^*$, let $a \in \bar{A}$, and suppose that $(G_w,\map_w)$ is defined such that $\prop{1}$ is satisfied for $w$.
Then $(G_{wa},\map_{wa})$ is built based on $(G_{w},\map_{w})$ in three steps.
First, we build a graph $G_{wa}'$ by adding to $G_w$ edges corresponding to the function $\map_w$.
Second, we build a graph $G_{wa}''$ by adding to $G_{wa}'$ the local transitions induced by the letter $a$.
Finally, we reduce $G_{wa}''$ into an element $G_{wa} = (V_{wa},E_{wa})$ of $\mathcal{F}_Q$.
\begin{itemize}
\item
Let $G_{wa}' = (V_{wa}',E_{wa}')$ be the graph defined by $V_{wa}' = V_w \cup Q^{\ins}$, where $Q^{\ins}$ is a copy of the set $Q$, and $E_{wa}' = E_w \cup \{(\map_w^{-1}(p),p^{\ins}) \in 2^Q \times Q^{\ins} | p \in \textup{Im}(\map_w) \}$.
Since \prop{1} is satisfied for $w$ by supposition, for every $(p,q) \in \bh_{rr}(w)$ there is a path in $G_{wa}'$ between $s_p$ and $q^{\ins}$.
\item
Let $G_{wa}'' = (V_{wa}'',E_{wa}'')$ be the graph defined by $V_{wa}'' = V_{wa}' \cup Q^{\outs}$, where $Q^{\outs}$  is a copy of the set $Q$, and let $E_{wa}'' = E_{wa}' \cup \{(p^{\ins},\tau(p)) \in Q^{\ins} \times V_{wa}' | p \in Q \},$
where
\[
\tau(p) = 
\left \{
\begin{array}{lll}
q^{\outs} \in Q^{\outs} & \textup{ if } \delta(p,a) = (q,+1),\\
q^{\ins} \in Q^{\ins} & \textup{ if } \delta(p,a) = (q,0),\\
s_q \in V_w & \textup{ if } \delta(p,a) = (q,-1) \textup{ and $s_q$ is the smallest vertex containing $q$}.
\end{array}
\right .
\]
For every $(p,q) \in \bh_{rr}(wa)$, there exists a path in $G_{wa}'$ between $p^{\ins}$ and $q^{\outs}$.
\item
Let $\inp(s) : V_{wa}'' \rightarrow 2^Q$ be the function mapping each vertex $s$ of $G_{wa}''$ to the set of states $q$ such that there exists a path from $q^{\ins}$ to $s$ in $G_{wa}''$.
Moreover, let $\outp(s) : V_{wa}'' \rightarrow Q \cup \{ \bot \}$ be the function mapping each vertex $s$ of $G_{wa}''$, to $\bot$ if the unique path starting from $s$ loops infinitely, and to $q \in Q$ if the target of this path is the vertex $q^{\outs} \in Q^{\outs}$.
Let $G_{wa} = (V_{wa},E_{wa})$, where 
$V_{wa} = \{ \inp(s) \subset Q | s \in V_{wa}'', \outp(s) \neq \bot \}$,
and let $\map_{wa}$ be the function mapping each vertex $\inp(s)$ of $V_{wa}$ to $\outp(s)$.
For every $(p,q) \in \bh_{rr}(wa)$, since $\outp(p^{\ins}) = q$, $\inp(p^{\ins})$ is a vertex of $G_{wa}$, and $\map_{wa}(\inp(p^{\ins})) = \outp(p^{\ins}) = q$, hence \prop{1} is satisfied, as $\inp(p^{\ins}) = s_p$.
\end{itemize}

\begin{remark}\label{remark:pres_rtr}
Since the construction of $G_{wa}$ only depends on $G_w$ and $a$, for every $w'\in \bar{A}^*$ such that $G_{w'} = G_{w}$, we have $G_{w'a} = G_{wa}$.
\end{remark}

\begin{remark}\label{remark:pres_mer}
By construction, for each subset $s$ of $Q$, $s$ is a vertex of the graph $G_{w}$ if and only the subset of right-to-right runs of $T$ over $w$ containing the runs whose starting state belongs to $s$ merge at some point, before merging with any other.
\end{remark}

The next results will allow us to obtain the bound over the size of the set of states presented in the statement of Theorem \ref{Thm:2wSSTb}.

\begin{lemma}\label{lemma:bound_vert}
Every rooted forest $G = (V,E)$ in $\mathcal{F}_Q$ has at most $2|Q|-1$ vertices.
\end{lemma}

\begin{proof}
This is proved by induction over $|Q|$.
If $|Q| = 1$, $|2^Q \setminus \emptyset| = 1$, hence $|V| \leq 1 = 2|Q| -1$.
Now suppose that $|Q| > 1$, and that the the result is true for every set $Q'$ such that $|Q'| < |Q|$.
Let $G'$ be the graph obtained by removing all the roots of $G$.
Then $G'$ is a union of trees $G_1 = (V_1,E_1), G_2 = (V_2,E_2), \ldots, G_m = (V_m,E_m)$.
For every $1 \leq i \leq m$, the forest $G_i$ is an element of $\mathcal{F}_{s_i}$, where $s_i \subset Q$ denotes the root of $G_i$.
Therefore, by using the induction hypothesis, we have
\[
\begin{array}{lll}
|V| & = & 1 + |V_1| + |V_2| + \ldots + |V_m|\\
& \leq & 1 + (2|s_1| - 1) + (2|s_2| - 1) + \ldots + (2|s_m| - 1)\\
& = & 2(|s_1| + \ldots +|s_m|) + 1 - m\\
& \leq & 2|Q| - 1.
\end{array}
\]
\end{proof}

\begin{lemma}\label{lemma:bound_size}
Let $n = |Q|$.
The size of $\mathcal{F}_Q$ is smaller than or equal to $\frac{(2n)^{2n-2}}{(n-2)!}$.
\end{lemma}

\begin{proof}
By Cayley's Formula, there exists exactly $(2n)^{2n-2}$ labeled trees on $2n$ vertices.
The result follows from the fact that any element of $\mathcal{F}_Q$ can be represented by a tree on $2n$ vertices, of which at most $n+2$ are labeled, which justifies the denominator, as the labels of $(n-2)$ vertices can be forgotten.

Given an element $G = (V,E)$ of $\mathcal{F}_Q$, we know that $|V| \leq 2n-1$ by Lemma \ref{lemma:bound_vert}.
Let $G'$ be the tree on $2n$ vertices obtained by adding to $G$ a vertex $s_\bot$, an edge from each root of $G$ to $s_\bot$, and, if $|V| < 2n-1$, a linear path composed of $2n-|V|-1$ new vertices starting from a vertex $s_{\top}$, and whose end is linked to $s_{\bot}$. 
Then $G$ can be computed back from any graph isomorphic to $G'$, as long as we set the label $\bot$ and $\top$ to the vertices corresponding to $s_{\bot}$ and $s_{\top}$, and the label $q$ to the vertex corresponding to $s_q$ , where $s_q$ is the smallest vertex of $V$ containing $q$, for every $q$ appearing in a vertex of $G$.
\end{proof}

\begin{corollary}\label{lemma:bound_size2}
Let $n = |Q|$.
The size of $\{ (G_w,\map_w)| w \in \bar{A}^* \}$ is smaller than or equal to $(2n)^{2n}$.
\end{corollary}

\begin{proof}
For every $w \in \bar{A}^*$, since $\map_w$ is an injective function mapping the trees of $G_w$ to $Q$, and $G_w$ contains at most $n$ trees,
\[
|\{ (G_w,\map_w)| w \in \bar{A}^* \}| \leq n! |\mathcal{F}_Q| \leq (2n)^{2n}
\]
\end{proof}

An other consequence of Lemma \ref{lemma:bound_vert} is that for every word $w \in \bar{A}^*$, the graph $G_w = (V_w,E_w)$ admits an injective vertex labeling $\lambda_w : V_w \rightarrow \mathcal{X}$, where $\mathcal{X} = (X_1, \ldots, X_{2|Q|-1})$ is a set containing $2|Q| - 1$ variables.
We shall now present for every $w \in \bar{A}^*$, the construction of a substitution $\sigma_w \in \mathcal{S}_{\mathcal{X},B}$ that will allow us, together with the graph $G_w$ and its vertex labeling $\lambda_w$, to describe the output production of the right-to-right runs of $T$ over $w$.
Formally, for every $(p,q) \in \bh_{rr}(w)$, let $w_{p,q} \in B^*$ denote the production of the corresponding right-to-right run.
Moreover, for every vertex $s$ of $G_w$, let $\bar{\lambda}_w(s)$ denote the concatenation of the $\lambda$-labels of the vertices forming the path starting from the $s$ in $G_w$, and for every state $q \in Q$ that appears at least in one vertex of $G_w$, let $s_q \in 2^Q$ be the vertex of $G_w$ of minimal size such that $q \in s_q$.
Then we want the following property to be satisfied.

\begin{description}
\item[\prop{2}]
For all $(p,q) \in \bh_{rr}(w)$, we have
$
(\sigma_{w})(\bar{\lambda}_w(s_p)) = w_{p,q}.
$
\end{description}

Let $\sigma_{\epsilon}$ be the substitution mapping each variable to $\epsilon$.
Once again, since $\bh_{rr}(w)$ is empty, \prop{2} is trivially satisfied for $w = \epsilon$.
Let $w \in \bar{A}^*$, let $a \in \bar{A}$, and let us suppose that \prop{2}  is satisfied for $w$.
The substitution $\sigma_{wa}$ is defined as the composition of a substitution $\sigma_{w,a}$, whose construction we will now present, with $\sigma_w$. 
In order to build $\sigma_{w,a}$, we define a vertex labeling $\mu_{wa} : V_{wa} \rightarrow (\mathcal{X} \cup B)^*$.
We require $\mu_{wa}$ to be copyless, i.e., for any variable $X \in \mathcal{X}$, there exists at most one vertex $s$ such that $X$ occurs in $\mu_{wa}(s)$.
Then, we define $\sigma_{w,a}$ as the copyless substitution mapping $\lambda_{wa}(s)$ to $\mu_{wa}(s)$.
The labeling $\mu_{wa}$ is obtained by first extending the labeling $\lambda_w$ of $V_w$ to a labeling $\mu''_{wa}$ of $V_{wa}''$, and then reducing it to $G_{wa}$.
We denote by $\bar{\mu}_{wa}(s)$ (resp. $\bar{\mu}_{wa}''(s)$) the concatenation of the labels of the vertices forming the path starting from a vertex $s$ of $G_{wa}$ (resp. $G_{wa}''$).

\begin{itemize}
\item
Let $\mu_{wa}'' : V_{wa}'' \rightarrow \mathcal{X} \cup B^*$ be the substitution mapping $s \in V_w$ to $\lambda_w(s)$, $q^{\outs} \in Q^{\outs}$ to $\epsilon$, and $q^{\ins} \in Q^{\ins}$ to $\gamma(p,a,q,m) \in B^*$, where $\delta(p,a) = (q,m)$.
Since $\lambda_w$ is injective, this labeling is copyless.
Moreover, since \prop{2} is satisfied for $w$ by supposition, by definition of $G_{wa}''$ we have, for every $(p,q) \in \bh_{rr}(wa)$, 
\[
\sigma_{w}(\bar{\mu}_{wa}''(p^{\ins})) = w_{p,q}.
\]
\item
For every $t \in V_{wa}$, the set of vertices $s$ of $G_{wa}''$ such that $\inp(s) = t$ is not empty, and they form a path $s_1, \ldots, s_m$.
Let $\mu(s) = \mu'(s_1) \ldots \mu'(s_m)$.
Since $\mu_{wa}''$ is copyless, so is $\mu_{wa}$.
Moreover,
\[
(\sigma_{w})(\bar{\mu}_{wa}(\inp(p^{\ins}))) = w_{p,q}.
\]
Since $\sigma_{wa} = \sigma_w \circ \sigma_{w,a}$, $\sigma_{w,a}(\lambda_{wa}(s)) = \mu_{wa}(s)$ by definition, and $\inp(p^{\ins}) = s_p$, this proves that \prop{2} is satisfied for $wa$.
\end{itemize}

\begin{remark}\label{remark:pres_pr}
Since the construction of $\sigma_{w,a}$ only depends on $G_w$ and $a$, for every $w'\in \bar{A}^*$ such that $G_{w'} = G_{w}$, we have $\sigma_{w',a} = \sigma_{w,a}$.
\end{remark}

We are now ready to define formally the copyless \sst $T'=(A,B,P,j,Q_f,\alpha,\mathcal{X},\beta,F')$.
\begin{itemize}
\item $P= \{  (G_{\vdash w},\map_{\vdash w}) | w \in A^* \}$,
\item $j = (G_{\vdash},\map_{\vdash})$,
\item $Q_f = \{ (G_{\vdash w},\map_{\vdash w}) | \map_{\vdash w \dashv}( i ) \in F \}$,
\item $\alpha: P \times A \rightarrow P, ((G_{\vdash w},\map_{\vdash w}),a) \mapsto (G_{\vdash wa},\map_{\vdash wa})$, which is well-defined, by Remark \ref{remark:pres_rtr},
\item $\mathcal{X}=\{X_i | 1 \leq i \leq 2|Q|-1\}$,
\item $\beta: P \times A \rightarrow \mathcal{S}_{\mathcal{X},B}, ((G_{\vdash w},\map_{\vdash w}),a) \mapsto \sigma_{\vdash w,a}$, which is well-defined, by Remark \ref{remark:pres_pr},
\item $F': Q_f \rightarrow (\mathcal{X} \cup B)^*$, $(G_{\vdash w},\map_{\vdash w})) \mapsto \sigma_{\vdash w \dashv}(\map_{\vdash w \dashv}(i))$.
\end{itemize}

The state reached by $T'$ over an input word $w \in A^*$ is $G_{\vdash w}$, and the substitution induced by the corresponding run is $\sigma_{\vdash w}$.
Therefore, by \prop{1} and the definition of $Q_f$, the domain of the functions defined by $T$ and $T'$ are identical.
Moreover, since $\sigma_{\vdash}$ is the substitution mapping all the variables to $\epsilon$, by supposition, we have, by \prop{2} and the definition of $F'$, that the image of a given word by those two functions are also identical, hence the functions are the same.

\TWtoSSTApb*

To prove this theorem, we introduce a new equivalence relation, and prove that the aperiodicity of the \TWDFT implies aperiodicity of this relation. The aperiodicity of the \sst then follows from this.
We say that two words $v$ and $w$ are \emph{merge equivalent} ($v\sim_m w$) if they induce the same merges in the same order in their four behavior relations.
Let us remark that if $v\sim_m w$, then for any word $u$, $G_{uv}=G_{uw}$.
This is due to the fact that $G_{uv}$ represents the merges of the right-to-right runs over $uv$, and these runs can be decomposed in right-to right runs over $u$ and partial runs over $v$.

\begin{lemma}\label{lemma-MergeOrder}
Let $w \in \bar{A}^*$ be such that $w^n \sim_T w^{n+1}$.
Then $w^n \sim_m w^{n+1}$.
\end{lemma}
\begin{proof}
Let $w$ and $n$ be such that $w^n \sim_T w^{n+1}$.
By definition, they model the same partial runs, and thus the same merges appear. Thus we only need to prove that they appear in the same order.
Consider two merges that appear consecutively in $w^n$.
We prove that they appear in the same order in $w^{n+1}$, depending on the kind of partial run that we consider.
If the merges are of right-to-right or left-to-left runs, then the exact same runs appear in $w^{n+1}$ and thus the same merges appear.
If they affect right-to-left (resp. left-to-right) runs, then the rightmost (resp. leftmost) $n$ iterations of $w$ in $w^{n+1}$ will merge these run. By noticing that after merging, two runs can not be separated again, the two merges will appear in the same order in $w^{n+1}$, concluding the proof.
\end{proof}

The two following lemmas now conclude the proof, since the aperiodicity of both the underlying automaton and the substitution imply the aperiodicity of the whole \sst.

\begin{lemma}\label{TWtocSST-ApSt}
Let $w \in \bar{A}^*$ be such that $w^n \sim_T w^{n+1}$.
Then for every $u \in \bar{A}^*$, $G_{uw^{n+1}} = G_{uw^{n}}$.
\end{lemma}
\begin{proof}
This lemma comes directly from Lemma~\ref{lemma-MergeOrder} and the previous remark stating that the merge equivalence implies the equivalence for the underlying automaton of the \sst.
\end{proof}

\begin{lemma}
Let $w \in \bar{A}^*$ be such that $w^n \sim_T w^{n+1}$.
Then for every $u \in \bar{A}^*$, $\sigma_{u,w^{n}} \sim \sigma_{u,w^{n+1}}$.
\end{lemma}
\begin{proof}
By construction, reducing a graph from the proof of Theorem~\ref{Thm:2wSSTb} $G'_{ua}$ to $G_{ua}$ amounts to delete the unnecessary information from $G_u$, i.e.
deleting cycles and reducing paths with no new merges to a single vertex.
Thus each vertex from $G_u$ can either be traced back to a vertex of $G_{ua}$ or is deleted.
Should we forget about the production, the flow of variables then corresponds exactly to this.
Thanks to Lemma~\ref{lemma-MergeOrder}, we know that $w^n$ and $w^{n+1}$ are merge equivalent.
Then given a state $G_u$ and one of its vertex, it can be traced back to the same vertex after reading  $w^n$ or $w^{n+1}$.
Since we have a unique way of mapping vertices to variables, the substitutions
 $\sigma_{uw^{n}}$ and $\sigma_{uw^{n+1}}$ will be equal when production is erased, proving the aperiodicity of the substitution function.
\end{proof}

\subsection*{From $k$-bounded to $1$-bounded \sst}

\KtoO*

\begin{proof}
In order to move from a $k$-bounded \sst to a $1$-bounded \sst, the natural idea is to use copies of each variable.
However, we cannot maintain $k$ copies of each variable all the time: suppose that $X$ flows into $Y$
and $Z$, which both occur in the final output. If we have $k$ copies of $X$, we cannot produce
in a $1$-bounded way $k$ copies of $Y$ and $k$ copies of $Z$. We will thus limit, for each variable $X$, the number of copies 
of $X$ we maintain. In order to get this information, we will use a look-ahead information
on the suffix of the run.

The proof relies on the following fact: suppose that we know at each step what is the substitution 
induced by the suffix of the run. From this substitution, for each variable $X$ we know
the value of the integer $n$ such that $X$ will be involved exactly $n$ times in the final output. We 
can thus copy each variable sufficiently many times and use them to produce this substitution 
in a copyless fashion.

One can observe that there are finitely many substitutions, this information being held in the transition monoid of the \sst.
Then we can compute, at each step and for each possible substitution, a copyless update.
But as a given element of the monoid may have several successors, the update function flows variables 
from one element to 
variables of several elements.
As these variables are never recombined, we get the $1$-boundedness of the construction.

Let $T=(A,B,Q,q_0,Q_f,\delta,\cX,\rho,F)$ be an aperiodic $k$-bounded \sst, $M_T$ be its transition
monoid, and $\eta_T:A^*\to M_T$ be its transition morphism.

We construct $T'=(A,B,Q',q'_0,Q'_f,\delta',\cX',\rho',F')$ where:
\begin{itemize}
\item The set of states $Q'=Q\times\cP(M_T)$ is the current state plus a set of elements of $M_T$ corresponding to the possible images of the current suffix.
\item $q_0'=(q_0,S_0)$ where $S_0=\{m\in M_T\mid \delta(q_0,m)\in Q_f\}$ is the set of relevant possible images of input words.
Here, we are abusing notations as $\delta(q,m)$ stands for $\delta(q,u)$ where $\eta_T(u)=m$. By definition of the transition monoid, we have that $\eta_T(u)=\eta_T(v)$ implies  $\delta(q,u)=\delta(q,v)$, thus this is well defined.
\item $Q'_f=\{ (q,S) \mid q\in F \text{ and }1_{M_T}\in S\}$.
\item $\delta': Q'\times A \to Q'$ is defined by $\delta((q,S),a)=(q',S')$ where $q'=\delta(q,a)$ and 
$S'=\{m\in M_T\mid \eta_T(a)m\in S \}$.
\item  $\cX'=\cX\times M_T \times \{1,\ldots,k\}$. Variables from $\cX'$ will be denoted $X_i^m$ for $X\in \cX, i\leq k$ and $m\in M_T$.
\item The variable update function is defined as follows.
First given a state $q$ of $T$ and an element $m$ of $M_T$, we define $\tilde{\sigma}_{q,m}$ as the projection of the output substitution induced by a run starting on $q$ on a word whose image is $m$, i.e $\tilde{\sigma}_{q,u}=\gamma(q,u)\circ F(\delta(q,u))$ for $\eta_T(u)=m$.
Note that by definition of the transition monoid of an \sst, it is well defined.

Now consider a transition $(q,S)\xrightarrow{a}(q',S')$, $n\in S'$ and $0<i\leq |\tilde{\sigma}_{q',n}|_X$, 
$\rho'((q,S),a,X_i^n)$ is defined similarly to $\rho(q,a,X)$, where all variables are labeled by the element $\eta_T(a)n$ and numbered to ensure the $1$-bounded property.
Such a numbering is possible thanks to the fact that $n$ indicates which variables are used as well as how many times.
This allows us to copy each variables the right amount of times, using different copies at each occurrence. The $k$-bounded property then ensures that we will never need more than $k$ variables for a possible output.

\item $F':Q'_f\to (B\cup \cX')^*$ is defined as follows. Let $(q,S)\in Q'_f$.
The string $F'(q,S)$ is obtained from the string $F(q)$
by substituting each variable $X$ by a variable $X^n_i$, where $0<i\leq |F(q)|_{X}$ and $n=1_{M_T}$.
\end{itemize}
\end{proof}

\KtoOAp*
\begin{proof}
We now have to prove that $T'$ is aperiodic.
 We claim that the runs of $T'$ are of the form
 $(q,S)\xrightarrow{u} (q',S')$ where $q\xrightarrow{u}_T q'$ and $S'=\{ n \in M_T \mid \eta_T(u)n\in S\}$, which holds by construction of $T'$.
 The update for such a run is then the update of $T$ over the run $q\xrightarrow{u}_T q'$, where variables are labeled by elements from $S$ and $S'$ and numbered accordingly.
 Then as $T$ is aperiodic, the $Q$ part of the run is also aperiodic by construction.
 The other part computes sets of runs according to $M_T$, which is also aperiodic.
 Then $T'$ will also be aperiodic as the set $S'$ only depends on the image of the word read, and by definition $\eta_T(u^n)=\eta_T(u^{n+1})$ for $n$ large enough.
\end{proof}

\end{document}